\documentclass[runningheads]{llncs}

\usepackage[T1]{fontenc}
\usepackage{acronym}
\usepackage[colorinlistoftodos]{todonotes}
\usepackage{amsmath, amssymb}
\usepackage{hyperref}
\usepackage[plain]{fancyref}
\usepackage{geometry}
\usepackage{graphicx}
\usepackage[inline]{enumitem}
\usepackage{xcolor}
\usepackage{listings}
\usepackage{setspace}
\usepackage{proof}
\usepackage{xspace}

\usepackage[backend=biber,
   style=lncs,
   citestyle=numeric,
   natbib=true,
   maxcitenames=1, 
   mincitenames=1,
   giveninits=true,
   hyperref=true,
   sorting=none,
   defernumbers]{biblatex}

\DeclareNameAlias{default}{family-given}
\AtEveryBibitem{%
   \clearlist{address}
   \clearfield{date}
   \clearfield{eprint}
   \clearfield{month}
   \clearfield{note}
   \clearfield{arxivId}
   \clearfield{volume}
   \clearfield{url}
   \clearname{editor}
   \ifentrytype{inproceedings}
     {\clearfield{day}
      \clearfield{month}
      \clearfield{volume}}{}
}

\definecolor{athenaKeyword}{HTML}{59007f}     
\definecolor{athenaOperator}{rgb}{0.2,0.2,0.7}  
\definecolor{athenaComment}{rgb}{0.3,0.6,0.3}   
\definecolor{athenaString}{rgb}{0.1,0.2,0.3}    

\definecolor{maudeKeyword}{HTML}{003f7f}     
\definecolor{maudeOperator}{rgb}{0.2,0.2,0.7}  
\definecolor{maudeComment}{rgb}{0.25,0.5,0.25} 
\definecolor{maudeString}{rgb}{0.55,0.15,0.15}  

\lstdefinelanguage{Athena}{
  morekeywords=[1]{else,holds?,error,module,domain,domains,declare,define,forall,exists,apply-method,%
    conclude,assume,suppose-absurd,pick-any,with-witness,pick-witness,%
    for,by-induction,datatype,datatypes,datatype-cases,structure,%
    check,match,let,letrec,try,assert,lambda,method},
  keywordstyle=[1]\color{athenaKeyword}\bfseries,
  morekeywords=[2]{|,\~,\=>,\==>,\<==>,=,:=},
  keywordstyle=[2]\color{athenaOperator}\bfseries,
  sensitive=true,
  morecomment=[l]{\#},
  commentstyle=\color{athenaComment}\ttfamily,
  morestring=[b]",
  stringstyle=\color{athenaString}\ttfamily,
  showstringspaces=false,
  alsoletter={-}
}

\lstdefinestyle{athena}{
  language=Athena,
  basicstyle=\ttfamily\small,
  tabsize=2,
  keepspaces=true,
  showstringspaces=false,
  breaklines=true,
  aboveskip=1em,
  belowskip=1em,
  numbers=left,
  numberstyle=\tiny\color{gray},
  captionpos=b,
  escapechar=`
}

\lstnewenvironment{athena}[1][]
 {\lstset{style=athena,#1}}{}

\lstdefinestyle{athenaInline}{
  language=Athena,
  basicstyle=\ttfamily\small,
  tabsize=2,
  keepspaces=true,
  showstringspaces=false,
  breaklines=true,
  aboveskip=1em,
  belowskip=1em,
  numbers=none,
  captionpos=b,
  escapechar=`
}
\lstnewenvironment{athenainline}[1][] {\lstset{style=athenaInline,#1}}{}

\lstdefinelanguage{Maude}{%
  morekeywords=[1]{ctor,mod,fmod,endm,endfm,pr,protecting,ex,extending,inc,including,sort,sorts,subsort,subsorts,mb,cmb,var,vars,op,ops,eq,ceq,rl,crl,if,search,red,reduce},
  sensitive=true,
  morecomment=[l]{--},
  morecomment=[s]{/*}{*/},
  morestring=[b]"%
}
\lstdefinestyle{maude}{
  language=Maude,
  basicstyle=\ttfamily\singlespacing\small,
  keywordstyle=\color{maudeKeyword}\bfseries,
  commentstyle=\color{maudeComment}\ttfamily,
  stringstyle=\color{maudeString}\ttfamily,
  columns=flexible,
  numbers=left,
  numberstyle=\tiny\color{gray},
  showstringspaces=false,
  keepspaces=true,
  aboveskip=-3pt,
  breaklines=true,
  tabsize=2,
  escapechar=`,
  literate={_}{\_}1, 
  captionpos=b        
}

\lstnewenvironment{maude}[1][] {\lstset{style=maude,#1}}{}
\lstnewenvironment{maudeinline}[1][] {\lstset{style=maude,#1, numbers=none}}{}

\newcommand{\athenacmd}[1]{\texttt{\textcolor{athenaKeyword}{\detokenize{#1}}}}
\newcommand{\athenacode}[1]{\texttt{\detokenize{#1}}}
\newcommand{\maudecmd}[1]{\texttt{\textcolor{maudeKeyword}{\detokenize{#1}}}}

\def\fref{\Fref}
\renewcommand{\tablename}{Table}  
\renewcommand{\figurename}{Figure}  
  
\newcommand*{\fancyreflnlabelprefix}{ln}
\newcommand*{\Freflnname}{Line}
\newcommand*{\freflnname}{\MakeLowercase{\Freflnname}}
\Frefformat{vario}{\fancyreflnlabelprefix}%
  {\Freflnname\fancyrefdefaultspacing#1}
\frefformat{vario}{\fancyreflnlabelprefix}%
  {\freflnname\fancyrefdefaultspacing#1}
\Frefformat{plain}{\fancyreflnlabelprefix}%
  {\Freflnname\fancyrefdefaultspacing#1}
\frefformat{plain}{\fancyreflnlabelprefix}%
  {\freflnname\fancyrefdefaultspacing#1}

\newcommand{\lnrange}[2]{Lines~\ref{#1}--\ref{#2}}    

\newcommand*{\fancyrefdeflabelprefix}{def}
\newcommand*{\Frefdefname}{Definition}
\newcommand*{\frefdefname}{\MakeLowercase{\Frefdefname}}
\Frefformat{vario}{\fancyrefdeflabelprefix}%
  {\Frefdefname\fancyrefdefaultspacing#1#3}
\frefformat{vario}{\fancyrefdeflabelprefix}%
  {\frefdefname\fancyrefdefaultspacing#1#3}
\Frefformat{plain}{\fancyrefdeflabelprefix}%
  {\Frefdefname\fancyrefdefaultspacing#1}
\frefformat{plain}{\fancyrefdeflabelprefix}%
  {\frefdefname\fancyrefdefaultspacing#1}    
  
\newcommand*{\fancyreftheolabelprefix}{theo}
\newcommand*{\Freftheoname}{Theorem}
\newcommand*{\freftheoname}{\MakeLowercase{\Freftheoname}}
\Frefformat{vario}{\fancyreftheolabelprefix}%
  {\Freftheoname\fancyrefdefaultspacing#1#3}
\frefformat{vario}{\fancyreftheolabelprefix}%
  {\freftheoname\fancyrefdefaultspacing#1#3}
\Frefformat{plain}{\fancyreftheolabelprefix}%
  {\Freftheoname\fancyrefdefaultspacing#1}
\frefformat{plain}{\fancyreftheolabelprefix}%
  {\freftheoname\fancyrefdefaultspacing#1}    

\newcommand*{\fancyrefproplabelprefix}{prop}
\newcommand*{\Frefpropname}{Proposition}
\newcommand*{\frefpropname}{\MakeLowercase{\Frefpropname}}
\Frefformat{vario}{\fancyrefproplabelprefix}%
  {\Frefpropname\fancyrefdefaultspacing#1#3}
\frefformat{vario}{\fancyrefproplabelprefix}%
  {\frefpropname\fancyrefdefaultspacing#1#3}
\Frefformat{plain}{\fancyrefproplabelprefix}%
  {\Frefpropname\fancyrefdefaultspacing#1}
\frefformat{plain}{\fancyrefproplabelprefix}%
  {\frefpropname\fancyrefdefaultspacing#1}    

\newcommand*{\fancyreflstlabelprefix}{lst}
\newcommand*{\Freflstname}{Listing}
\newcommand*{\freflstname}{\MakeLowercase{\Freflstname}}
\Frefformat{vario}{\fancyreflstlabelprefix}%
  {\Freflstname\fancyrefdefaultspacing#1#3}
\frefformat{vario}{\fancyreflstlabelprefix}%
  {\freflstname\fancyrefdefaultspacing#1#3}
\Frefformat{plain}{\fancyreflstlabelprefix}%
  {\Freflstname\fancyrefdefaultspacing#1}
\frefformat{plain}{\fancyreflstlabelprefix}%
  {\freflstname\fancyrefdefaultspacing#1}

\newcommand*{\fancyrefremlabelprefix}{rem}
\newcommand*{\Frefremname}{Remark}
\newcommand*{\frefremname}{\MakeLowercase{\Frefremname}}
\Frefformat{vario}{\fancyrefremlabelprefix}%
  {\Frefremname\fancyrefdefaultspacing#1#3}
\frefformat{vario}{\fancyrefremlabelprefix}%
  {\frefremname\fancyrefdefaultspacing#1#3}
\Frefformat{plain}{\fancyrefremlabelprefix}%
  {\Frefremname\fancyrefdefaultspacing#1}
\frefformat{plain}{\fancyrefremlabelprefix}%
  {\frefremname\fancyrefdefaultspacing#1}

\definecolor{author}{rgb}{.5, .5, .5}
\definecolor{comment}{rgb}{.1, .0, .9}
\definecolor{note}{rgb}{.9, .4, .0}
\definecolor{idea}{rgb}{.1, .7, .0}
\definecolor{missing}{rgb}{.9, .1, .0}

{}

\newcommand{\mtoa}{$\textsf{maude2athena}$\xspace}
\newcommand{\eg}{\emph{e.g.,}\xspace}
\newcommand{\ie}{\emph{i.e.,}\xspace}

\newcommand{\tr}{\emph{tr}\xspace}
\newcommand{\ecal}{\mathcal{E}\xspace}

\acrodef{RL}{Rewriting Logic}
\acrodef{EC}{equivalence classes}
\acrodef{RT}{rewrite theory}
\acrodef{ITP}{Inductive Theorem Prover}

\begin{document}

\title{Equational and Inductive Reasoning for Maude in Athena}

\author{
    Mateo Sanabria\inst{1}\orcidID{0000-0003-3407-9792} \and
    Carlos Varela\inst{2}\orcidID{0000-0003-4708-3109} \and
    Camilo Rocha\inst{3}\orcidID{0000-0003-4356-7704} \and
    Nicolás Cardozo\inst{1}\orcidID{0000-0002-1094-9952}}

\authorrunning{M.Sanabria et al.}

\institute{Universidad de los Andes, Bogotá , Colombia \and
           Rensselaer Polytechnic Institute, Troy, USA \and
           Pontificia Universidad Javeriana, Cali, Colombia
}
\maketitle

\begin{abstract}
  In the rewriting logic framework, equational-based specifications
  are used to define deterministic functional behavior, abstract data
  types, and canonical representations of data. These specifications
  include a (possibly order-sorted) signature and equations
  interpreted modulo structural axioms, such as associativity,
  commutativity, and identity.  While equational rewriting provides a
  powerful basis for execution and symbolic reasoning, it does not by
  itself offer native support for inductive or deductive reasoning.
  This paper presents \mtoa, a framework that systematically
  translates Maude's equational theories into Athena, a
  theorem proving language designed to support natural deduction
  proofs over many-sorted first-order logic specifications, including inductive reasoning, equational chaining, case-based reasoning, and proofs by contradiction.
  The translation supports induction-based reasoning modulo structural axioms with parametric
  induction rules; it faithfully encodes membership equational logic
  in a many-sorted setting without exponential blowup under reasonable
  conditions.
  This approach preserves the semantics of the original specification,
  while ensuring that the translation remains compact and amenable to deductive reasoning.  This work helps bridge the gap between model checking and theorem proving, enabling formal verification efforts that can benefit from both of these approaches.

\end{abstract}


\section{Introduction}
\label{sec:introduction}

Rewriting Logic~\cite{meseguer1992conditional} has established itself
as a powerful semantic and specification framework for modeling
computational systems, supporting both executable specifications and
rigorous equational reasoning. Within this framework,
\emph{equational-based specifications} play a central role in defining
deterministic functional behavior, abstract data types, and canonical
representations of data. Languages such as Maude~\cite{clavel2007all}
provide strong support for this style of specification through
membership equational logic, offering subsort relations, operator
overloading, and rewriting modulo structural axioms such as
associativity, commutativity, and identity.
While rewriting logic is a logical framework that can be used for
symbolic reasoning, including reasoning over its own
specifications~\cite{clavel-reflection-tcs2007}, Maude currently
offers limited support for interactive theorem proving, especially for
proofs that require explicit inductive reasoning or fine-grained proof
control. In contrast, interactive theorem provers such as
Athena~\cite{arkoudas17} provide robust environments for
natural-deduction reasoning in many-sorted first-order logic, together
with strong support for induction and proof structuring. However,
Athena lacks native support for the order-sorted features that are
essential to many Maude specifications. As a consequence,
fully-fledged Maude specifications cannot be directly reused within
Athena without a careful reconciliation of their underlying logical
foundations.

This paper presents $\mtoa$, a framework that systematically
translates Maude functional modules into Athena modules, enabling
equational and inductive reasoning over Maude specifications within
Athena’s interactive proof environment. The translation is grounded in
a semantics-preserving mapping from membership equational logic to
many-sorted first-order logic, in which subsort relations are made
explicit through cast operators and coherence
axioms~\cite{li2018method}.  To ensure scalability and avoid
exponential blowup, the approach relies on the notion of strictly
sensible order-sorted algebras~\cite{li2018method,goguen1992order},
which guarantees a linear and unambiguous translation in the presence
of operator overloading.

A central challenge in this translation is the loss of native
structural induction when order-sorted datatypes are flattened into
many-sorted domains. Although this transformation preserves equational
reasoning, it removes the inductive structure that is essential for
proving many correctness properties. To address this issue, Athena is
extended with \emph{parametric structural induction primitives} that
reconstruct the induction principles implicit in the original Maude
specification. This design allows different induction schemes to be
instantiated as needed, enabling inductive proofs over translated
domains without resorting to encodings of operational semantics or
meta-level reasoning. In this paper, this parametric framework is
instantiated with structural induction for sufficiently complete
equational specifications~\cite{rocha-suffcompl-lpar2010}; other induction principles can be
accommodated as well, such as structural induction over the oriented
equations of a terminating equational theory.

By bridging executable equational specifications in Maude with
interactive inductive proofs in Athena, $\mtoa$ enables a
compositional workflow in which specification, execution, and proof
coexist within a single formal development. The proposed translation
preserves the equational theory of the original specification while
remaining compact and amenable to both interactive and automated
reasoning. The approach is validated through theoretical results
establishing the correctness of the translation, as well as through a
non-trivial case study involving the verification of a compiler
specification that relies heavily on subsorting and inductive
reasoning.\footnote{The $\mtoa$ framework is available at \url{https://github.com/FLAGlab/Maude2Athena}.}


\section{Preliminaries}
\label{sec:background}

This section presents an overview on rewriting logic in  Maude, and theorem proving in
Athena.

\subsection{Rewriting Logic and Maude}
\label{sec:maude}

Rewriting logic~\cite{meseguer1992conditional} is built on top of
membership equational logic~\cite{bouhoula2000specification} as the
underlying equational formalism.
Maude~\cite{clavel2007all} is a high-performance logical framework and
specification language based on rewriting logic.
Membership equational logic supports sorts, equations, and membership
axioms that characterize the elements of a sort.  While the full
language extends these characteristics with rewrite rules to represent
transitions in concurrent systems, this paper focuses on the
translation of the functional (equational) sub-language.

Membership equational logic is defined over a many-kinded signature
$(K, F)$, where $K$ is a set of kinds and $F = \{ F_{w,k} \}_{(w,k)
  \in K^* \times K}$ is a family of function symbols typed over those
kinds. In addition to kinds, the logic considers a $K$-indexed family
of sorts $S=\{S_k\}_{k \in K}$, which allows for ascribing sorts to terms
using the notation $t:s$ (term $t$ has sort $s$). Thus, a complete
signature in membership equational logic is a triple $\Sigma =
(K,F,S)$.
The sorts $S$ are equipped with a partial order $\le$ (subsorting),
and operators in $F$ are defined over these sorts. This view 
associates the kinds $K$ as the connected components of the subsort
relation $\le$. Since an order-sorted algebra can be canonically
mapped to a membership
algebra~\cite{meseguer1992conditional,membership_algebra}, the
order-sorted setting is used throughout this paper.  A membership
equational theory is a pair $(\Sigma, E)$, combining a signature
$\Sigma$ with a set of sentences $E$ (conditional equations and
membership axioms). Equality modulo $E$, denoted $=_E$, is the finest
congruence on the term algebra $T_\Sigma(X)$ that satisfies the Horn
clauses in $E$. That is, for any terms $t, u \in T_{\Sigma}(X)$, $t
=_E u$ iff $E \vdash t = u$.
Semantic structures of equational theories are called
\textit{algebras}. The expression $T_\Sigma$ is the \textit{sorted
  ground term algebra} of $\Sigma$ and $T_\Sigma(X)$ the
\textit{sorted term algebra} of $\Sigma$ with variables in
$X$. Likewise, the quotient structure $T_{\Sigma/E}$ is the
\textit{initial algebra} of $(\Sigma, E)$ meaning that two terms $t, u
\in T_\Sigma$ are in the same equivalence class iff $t =_E u$. This is
naturally extended to the algebra $T_{\Sigma/E}(X)$ of terms with
variables in $X$ modulo the equations $E$.

Maude partitions the specification into a set of structural axioms $A$
(such as associativity, commutativity, and identity) and a set of
functional equations $E'$.  Consequently, an equational theory is
understood as $(\Sigma, E' \cup A)$, where rewriting is performed
modulo the structural axioms $A$ by orienting the equations from left
to right.
In Maude, functional modules correspond to membership equational
theories and define data types and their operations. Equations are
treated as simplification rules applied left to right. By assuming
\textit{sort-decreasingness}, \textit{ground confluence}, and
\textit{operational termination}, repeated application of the oriented
equations reduces a term to its
canonical form~\cite{duran2012church,lucas2009operational}.
Furthermore, it is assumed that an equational theory $(\Sigma, E' \cup
A)$ is \textit{sufficiently complete} with respect to a subsignature
$\Omega \subseteq \Sigma$ meaning that, for any sort $s \in S$ and
term $t \in T_{\Sigma, s}$, there is a term $u \in T_{\Omega, s}$ such
that $t =_E u$.



\fref{lst:Peano_maude} presents a functional module specification of
Peano numbers. The \texttt{PEANO} module is delimited by the Maude
keywords \maudecmd{fmod} and \maudecmd{endfm}. Sorts are declared with
the \maudecmd{sort} keyword, and subsort relations use
\maudecmd{subsorts} together with the symbol \maudecmd{<}
(\lnrange{ln:sorts-start}{ln:sorts-end}). Operators are declared with
\maudecmd{op}; each declaration specifies the operator name, its sort
signature, and an optional set of attributes that declare the
structural set of axioms for the equational theory. Particularly,
the attribute \maudecmd{ctor} marks constructors for the sort. The
elements of the sort \maudecmd{Even} are defined recursively by a
conditional equation (\fref{ln:peano-cmb}), and the Peano axioms are
introduced as equations (\lnrange{ln:peano-eq-star}{ln:peano-eq-end})
defining the set $E$ for the equational theory.

\begin{lstlisting}[language=Maude, style=maude,
                 caption={Definition of Peano numbers in Maude.},
                 label={lst:Peano_maude}]
fmod PEANO is
	sort NzNat Even Nat . `\label{ln:sorts-start}`
	subsorts NzNat Even < Nat . `\label{ln:sorts-end}`
	op zero : -> Even [ctor] .
	op s_ : Nat -> NzNat [ctor] .
	op _+_ : Nat Nat -> Nat .
	vars N M : Nat .
	cmb s s N : Even if N : Even . `\label{ln:peano-cmb}`
	eq N + zero = N . `\label{ln:peano-eq-star}`
	eq N + s M = s (N + M) . `\label{ln:peano-eq-end}`
endfm
\end{lstlisting}

\subsection{Theorem Proving in Athena}
\label{sec:athena}

Athena~\cite{arkoudas17} is a dual deduction and computation language:  users can define
{\it procedures}, which abstract over computations (i.e., lambda calculus abstractions in
functional programming); and {\it methods}, which abstract over Fitch-style natural
deductions~\cite{arkoudas2000denotational,arkoudas2005}.  When evaluating procedures,
if successful, Athena can return values (of different types), or else, it can diverge or
result in an error.  When evaluating methods, if successful, Athena produces logically
sound theorems---modulo its assumption base---as sentences in many-sorted first order
logic~\cite{juola1998maria}.  Athena has been effectively used to reason about concurrent
programming---particularly, the actor model~\cite{musser-varela-agere-2013,varela-lncs-2026},
and about safety-critical cyber-physical systems~\cite{PAUL2025103184,paul_2023_aesm}.

Athena's proofs have a natural deduction
style. Logical sentences can
either be asserted into the assumption base or proven as theorems
using Athena's deduction tools, with a soundness guarantee that any
proven theorem is a logical consequence of sentences in the assumption
base.
Athena performs automatic sort-checking to prevent ill-sorted
expressions in specifications and allows theorems to be introduced at
an abstract level by encapsulating proofs in parameterized methods,
which can be instantiated to prove different specializations of the
abstract theorems.
%
%
It defines two syntactic categories: an \textit{expression}, which
represents a computation, and a \textit{deduction}, a logical argument
such as a proof.  A valid deduction always concludes with a
\textit{sentence}, \ie the conclusion of the proof.

Sorts form the foundation of Athena's deductive language, classifying
all terms in its sentences. Sorts can be defined as \textit{domains},
\textit{structures}, and \textit{datatypes}, each providing different
mechanisms for constructing elements and each assuming different primitive axioms.

\textit{Domains} are sorts that describe the sets of objects to be
modeled. A \textit{datatype} is a special kind of domain that is
\textit{inductively generated}, meaning that every element of the
domain can be built up in a finite number of steps by applying
\textit{constructors} of the datatype. Beyond providing syntactic
convenience over domains, datatypes can be assumed to follow
\textit{free-generation axioms} for their elements, ensuring that
different constructor applications create different elements
(\emph{no-confusion} axioms), and that every element is represented by
some constructor application (\emph{no-junk} axioms).

\textit{Structures}, just like regular datatypes, are inductively
generated by their constructors. The only difference is that the
constructors might not be injective, so that the same constructor
applied to two distinct sequences of arguments might result in the
same value.

The keyword \athenacode{domain} introduces domains in
Athena. \fref{lst:sort_example} presents a definition for Peano numbers based
on three domains: \athenacode{Nat, Even, NzNat}, with function
symbols \athenacode{zero,S,plus}, defined using the
\athenacode{declare} instruction. Note that, alternatively, Peano numbers can be defined using 
datatypes as: \athenacode{datatype Nat := zero | (s Nat)}.

\begin{athena}[caption={Peano example definition using domains.}, label={lst:sort_example}]
domains Nat, Even, NzNat
declare zero:[]        -> Even
declare s:   [Nat]     -> NzNat
declare plus:[Nat Nat] -> Nat [+] 
define  [n m]           := [?n:Nat ?m:Nat]
assert  Plus-zero-axiom := (forall n .  (zero + n) = n) `\label{ln:athena-peano-axioms-init} `
assert* Plus-s-axiom    := ((n + (s m)) = (s (n + m))) `\label{ln:athena-peano-axioms-end} `
\end{athena}

The datatype definition provides free-generation axioms, ensuring that
different constructor applications yield distinct elements. In
contrast, the domain based approach offers more expressiveness by
allowing sort refinements, like \athenacode{Even} and
\athenacode{NzNat}, which enable more precise sort specifications.
The version using domains is inspired by languages with subsort
capabilities, such as Maude, and leverages Athena's many-sorted
foundation to provide fine-grained sort control.

The semantics of Athena defines how phrases \textit{F} (expressions or deductions) are evaluated within an
\textit{evaluation context}, a four-tuple $\langle \rho, \beta, \sigma, \gamma
\rangle$ representing the logical and computational environment in
which expressions are interpreted, where:

\begin{itemize}
\item
  $\rho$ is the execution environment, a computable function that maps
  any given identifier $I$ either to a value $V$ or to a special
  \textit{unbound} token.

\item $\beta$ is the assumption base, a finite set of sentences.

\item $\sigma$ is the store, a computable function that maps any
  natural number (representing a memory location) to a value (the
  location's contents) or to a special \textit{unassigned} token.

\item $\gamma$ is a set of symbols with their associated signatures,
  and a collection of sort constructors (with arities). $\gamma$ also
  includes information on whether a given sort constructor is a
  datatype or structure, and if so, which function symbols are its
  constructors.
\end{itemize}

The result of evaluating a phrase \textit{F} in an evaluation context 
$\langle\rho,\beta,\sigma,\gamma\rangle$ is one of the following:

\begin{enumerate}
\item A pair $(V,\sigma')$ consisting of a value \textit{V} and a
  store $\sigma'$, where \textit{V} is the output of the evaluation
  and $\sigma'$ reflects any side effects accumulated during the
  evaluation.

\item A pair consisting of an error message and a store $\sigma'$,
  indicating the occurrence of an error during the computation.

  \item Nontermination.
\end{enumerate}

Sentences are added into the assumption base using \athenacmd{assert} or 
\athenacmd{assert*}, where \athenacmd{assert*} uses implicit universal quantification for all free 
variables, before inserting the sentence into the assumption base. 
\lnrange{ln:athena-peano-axioms-init}{ln:athena-peano-axioms-end}
in \fref{lst:sort_example} add the properties of Peano numbers into
the assumption base.

Once sentences are introduced into the assumption base, they can be
used in deductive reasoning. Further, the datatypes and structures
definitions allow deductions of theorems that are not derivable from
the free-generation axioms alone, by using the
\athenacode{by-induction} method. For example, assuming that
\athenacode{Nat} follows a datatype definition, \fref{lst:deduction_induction} displays
the induction proof for the associativity property for the plus
operator, \fref{ln:plus-assoc}.

In general, the method \athenacode{by-induction} requires the proof of
the property for all base cases (\textit{zero} in the case of Peano,
\lnrange{ln:zero-star}{ln:zero-end}) and the proof for all inductive
definitions (\textit{S} for \athenacode{Nat} datatype,
\lnrange{ln:s-start}{ln:chain1-end}). Additionally, this induction proof
showcases one of the most useful Athena methods: \athenacode{chain}
(\lnrange{ln:chain0}{ln:zero-end} and
\lnrange{ln:chain1}{ln:chain1-end}). \athenacode{chain} allows equational
chaining based proofs, where each step is labeled with its
justification.

\begin{athena}[caption={Deduction of associativity property using datatype induction.},
               label={lst:deduction_induction}]
define [p q r] := [?p:Nat ?q:Nat ?r:Nat]
define plus_associative := (forall p q r . ((q + r) + p) = (q + (r + p))) `\label{ln:plus-assoc}`
conclude plus_associative
by-induction plus_associative {
   zero =>  pick-any q r `\label{ln:zero-star}`
           (!chain [((q + r) + zero) `\label{ln:chain0}`
                   --> (q + r)          [Plus-zero-axiom]
                   <-- (q + (r + zero)) [Plus-zero-axiom]) `\label{ln:zero-end}`
   | (s p) => let {induction-hypothesis := `\label{ln:s-start}`
               (forall ?q ?r . (?q + ?r) + p = ?q + (?r + p))}
           pick-any q r
           (!chain [ ((q + r) + (s p))`\label{ln:chain1}`
                   --> (s ((q + r) + p))   [Plus-s-axiom]
                   --> (s ((q + (r + p)))) [induction-hypothesis]
                   <-- (q + (s (r + p)))   [Plus-s-axiom]
                   <-- (q + (r + (s p))) `\label{ln:chain1-end}` ])
} 
\end{athena}


\section{From Maude to Athena}
\label{sec:translation}

At the heart of the proposed approach, lies the need to translate the
order-sorted structure of terms of an equational theory in Maude to
the many-sorted setting of equational theories in Athena.
This section presents such a translation as a function from Maude
modules to Athena theories.
Translation of membership (of terms in sorts) is
accomplished by using predicates in Athena.

\subsection{Strictly Sensible Order-Sorted Signatures}
\label{sec:strictly_sensible}

The translation provides a linear translation of \emph{strictly sensible} algebras to
many-sorted algebras~\cite{li2018method}.
The key idea of the translation is to add an equivalence relation
called core equality to the translated many-sorted structures. By
defining this relation, the complexity of translating a strictly
sensible order-sorted algebra to a many-sorted one is reduced and the
translated many-sorted algebra equations only increase by a very small
amount of new equations. Algebras definitions are lifted to 
signatures to keep the presentation of their results concise.

Fix a membership signature $\Sigma = (K, F, S)$, with kinds $K$,
function symbols $F$, and sorts $S$, and a poset $(S, \leq)$ over the
sorts. For a function symbol $g : s_1 \times \cdots \times s_n \to s$,
$\mathrm{target}(g) = s$ denotes its result sort.
In what follows, assume $f,f',f'' \in F$ and, for a fixed $n \in
\mathbb{N}$ and $s_i, s_i', s_i'' \in S$ for $0 \leq i < n$, where
$
    f: s_{0} \times \dots \times s_{n} \rightarrow s \ ; \
    f': s'_{0} \times \dots \times s'_{n} \rightarrow s' \ ; \
    f'': s''_{0} \times \dots \times s''_{n} \rightarrow s''.
$

Argument compatibility captures the idea of overloaded function
symbols that can be treated uniformly as their pair-wise argument
positions share a common supersort. Strongly sensible signatures have
argument compatible function symbols in which the target
sorts are the same.

\begin{definition}\label{def:strong_sensible}
  Two overloaded function symbols $f$ and $f'$ are called
  \textit{argument compatible}, denoted $ac(f,f')$, iff $(\forall \, i
  \;|\; 0 \leq i \leq n \; : s_i \equiv_{\leq} s'_{i} )$, where
  $s_i\equiv_{\leq}s'_i$ is a predicate expressing that $s_i$ and
  $s'_i$ are in the same connected component of $(\Sigma, \leq)$.
%
  The signature $\Sigma$ is \textit{strongly sensible} iff $(\forall \, f,f'
  \;|\; ac(f,f') \; : s = s')$.
\end{definition}

On the other hand, maximal argument-bounding signatures ensure that
every function symbol has a representative among its compatible function symbols,
thereby defining one representative for each class of
argument compatible function symbols.

\begin{definition}
  \label{def:maxiam_argument_bounding}
  The signature $\Sigma$ is \textit{maximal argument-bounding} iff:
  \[
  (\forall \, f \;|\;(\exists \, f'' \;|\;(\forall \, f' \;|\; ac(f,f')\;: (\forall \, i  \;|\; 0 \leq i \leq n \; : s'_{i} \le s''_{i} ) ))).
  \]
\end{definition}

A \textit{strictly sensible} signature is a signature that is both
strongly sensible and maximal argument-bounding. 
Strictly sensible signatures play a central role in the translation
from order-sorted equational theories in Maude to many-sorted theories
in Athena. It rules out ambiguous cases of function symbol overloading and
ensures that the translation remains well-defined and systematic.  In
contrast to more general translations (e.g., ~\cite{meseguerTranslation}),
which can produce non-linear or ambiguous mappings in the presence of
unrestricted overloading, requiring strict sensibility guarantees that
every function symbol has a single, consistent representative in the
translated signature.
In particular, strict sensibility ensures that each function symbol has a
unique representative in the translated signature, allowing the
order-sorted algebra to be flattened into a many-sorted representation
without loss of soundness or consistency.

In the rest of the paper, it is assumed that any signature $\Sigma$ is strictly
sensible.

\subsection{The Translation Function}

Let $\ecal = (\Sigma, E \cup A)$ be a membership equational theory, with
signature $(K, F, S)$ and poset of sorts~$(S, \leq)$, and $\Omega
\subseteq \Sigma$ be a constructor subsignature for $\ecal$.
This section defines the mapping~$\ecal~\mapsto~(\beta, \gamma)$ as
the function $\tr$, where $\beta$ is an Athena assumption base and
$\gamma$ is its symbol set. The $\rho$ execution environment and $\sigma$ store are omitted
from the mapping because they are orthogonal to the generated Athena
representation.
The behavior of $\tr$ is driven primarily by the subsort relation
$\le$ and by the constructors $\Omega$.

\begin{definition}[Translation Function $\mathrm{tr}$]
\label{def:translation}
  Let $\ecal = (\Sigma, E \cup A)$ be a membership equational theory with
  signature $\Sigma = (K, F, S)$, sort poset $(S, \leq)$, and constructor
  subsignature $\Omega \subseteq \Sigma$. The function
  $\mathrm{tr} : \ecal \to (\beta, \gamma)$ is defined by the following
  five components:
  \begin{enumerate}[label=\textbf{(\roman*)}]
  \item \textbf{Sort translation ($\mathrm{tr}_S$).} For each sort $s \in S$:
    \[
    \mathrm{tr}_S(s) = \begin{cases}
      \mathtt{datatype}(s,\,\{c \in \Omega \mid \mathrm{target}(c) = s\})
        & \text{if } s \notin \mathrm{dom}(\leq) \cup \mathrm{ran}(\leq)
          \text{ and } \exists\, c \in \Omega.\; \\[4pt]
      \mathtt{domain}(s) & \text{otherwise}
    \end{cases}
    \]

  \item \textbf{Function symbol translation ($\mathrm{tr}_F$).}
    Let $[f]_{\mathrm{ac}}$ denote the argument compatibility class of
    $f \in F$. For each class $[f]_{\mathrm{ac}}$, select the maximal
    representative $f^{*}$ (guaranteed to exist by strict sensibility),
    where $f^{*}\!: s_1 \times \cdots \times s_n \to s$, and set
    \[
      \gamma \;\cup\;
      \bigl\{\, f^{*}\!:\![\mathrm{tr}_S(s_1) \cdots
      \mathrm{tr}_S(s_n)] \to \mathrm{tr}_S(s) \,\bigr\}
      \;\mapsto\; \gamma.
    \]
    For each $s \leq s'$ in $(S,\leq)$, set
    \[
      \gamma \;\cup\;
      \bigl\{\, \mathrm{Cast}_{s \to s'}\!:\![\mathrm{tr}_S(s)]
      \to \mathrm{tr}_S(s') \,\bigr\}
      \;\mapsto\; \gamma.
    \]
    Additionally, for each $f^{*}\!: s \times s \to s$ in $\gamma$
    with structural attributes in $A$, let $x, y, z$ be fresh
    variables of sort $\mathrm{tr}_S(s)$ and set
    \[
      \beta \;\cup\;
      \left\{\;
      \begin{aligned}
        &f^{*}(f^{*}(x,y),z) = f^{*}(x,f^{*}(y,z))
          &&\text{if }\mathtt{assoc} \in \mathrm{attr}(f), \\
        &f^{*}(x,y) = f^{*}(y,x)
          &&\text{if }\mathtt{comm} \in \mathrm{attr}(f), \\
        &f^{*}(\mathrm{tr}_T(e),x) = x,\;
         f^{*}(x,\mathrm{tr}_T(e)) = x
          &&\text{if }\mathtt{id}\!:\!e \in \mathrm{attr}(f)
      \end{aligned}
      \;\right\}
      \;\mapsto\; \beta.
    \]

  \item \textbf{Term translation ($\mathrm{tr}_T$).} For each
    $t \in T_{\Sigma}(X)$:
    \[
    \mathrm{tr}_T(t) = \begin{cases}
      x & \text{if } t = x \in X \\[2pt]
      f^{*}(\mathrm{tr}_T(t_1),\ldots,\mathrm{tr}_T(t_n))
        & \text{if } t = f(t_1,\ldots,t_n)
    \end{cases}
    \]
    where $f^{*}$ is the maximal representative of $[f]_{\mathrm{ac}}$
    and, at each argument position~$i$ where the subterm has sort $s_i'$
    with $s_i' \leq s_i$ and $s_i' \neq s_i$, $\mathrm{tr}_T(t_i)$ is
    wrapped in the appropriate cast chain
    $\mathrm{Cast}_{s_{j} \to s_{j+1}}(\cdots(\mathrm{tr}_T(t_i))\cdots)$
    so that the result has sort $\mathrm{tr}_S(s_i)$.

  \item \textbf{Equation translation ($\mathrm{tr}_E$).}
    For each equation $l = r$ in $E \cup A$, set
    \[
      \beta \;\cup\;
      \bigl\{\, \mathrm{tr}_T(l) = \mathrm{tr}_T(r) \,\bigr\}
      \;\mapsto\; \beta.
    \]
    For each composable pair $s \leq s' \leq s''$, add the
    \emph{core equality}:
    \[
      \beta \;\cup\;
      \bigl\{\, \mathrm{Cast}_{s' \to s''}(\mathrm{Cast}_{s \to s'}(x))
      = \mathrm{Cast}_{s \to s''}(x) \,\bigr\}
      \;\mapsto\; \beta,
    \]
    ensuring all casting paths between the same source and target sorts are identified.

  \item \textbf{Membership translation ($\mathrm{tr}_M$).}
    For each sort $s$ appearing in a membership axiom, set
    \[
      \gamma \;\cup\;
      \bigl\{\, \mathrm{is\_s}\!:\![\mathrm{tr}_S(\mathrm{kind}(s))] \to
      \mathtt{Boolean} \,\bigr\}
      \;\mapsto\; \gamma.
    \]
    Then, for each unconditional membership $t : s$ and each conditional
    membership $t : s \text{ if } P$, set
    $\beta \cup \{\mathrm{tr}_M(\cdot)\} \mapsto \beta$, where:
    \[
    \mathrm{tr}_M(t : s) = \mathrm{is\_s}(\mathrm{tr}_T(t))
    \qquad
    \mathrm{tr}_M(t : s \text{ if } P) = (\mathrm{tr}_T(P) \Rightarrow
      \mathrm{is\_s}(\mathrm{tr}_T(t)))
    \]
  \end{enumerate}
\end{definition}

The five components of $\tr$ interact as a pipeline: $\mathrm{tr}_S$
determines how each sort is represented in Athena;
$\mathrm{tr}_F$ uses these representations to declare function and cast
symbols and to assert the structural axioms of $A$;
$\mathrm{tr}_T$ relies on both to translate terms, inserting
casts where implicit subsorting must become explicit;
$\mathrm{tr}_E$ and $\mathrm{tr}_M$ then lift the term translation to
equations and membership axioms, respectively.
To make each step concrete, the rest of this section walks through
$\tr$ component by component on two variants of the Peano numbers
specification: the simple module in \fref{lst:peano_maude_simple}
(no subsort relation) and the richer module in
\fref{lst:Peano_maude} (with subsorts \maudecmd{NzNat} and
\maudecmd{Even}), showing the corresponding Athena output in
\fref{lst:peano_athena_simple} and \fref{lst:peano_athena}.

\begin{figure}[h]
\centering
\begin{minipage}[b]{0.48\textwidth}
\begin{lstlisting}[language=Maude, style=maude, caption={Simple Peano module in Maude.},
                 label={lst:peano_maude_simple}]
fmod PEANO is
    sort Nat . 
    op zero : -> Nat [ctor] .
    op s_ : Nat -> Nat [ctor] .
    op _+_ : Nat Nat -> Nat .
    vars N M : Nat .
    eq N + zero = N . 
    eq N + s M = s (N + M) . 
endfm
\end{lstlisting}
\end{minipage}
\hfill
\begin{minipage}[b]{0.48\textwidth}
\begin{lstlisting}[language=Athena, style=athena, caption={Simple Peano module in Athena.},
                 label={lst:peano_athena_simple}]
module PEANO {                       
    datatype Nat := zero  | (s Nat) `\label{ln:datatype_peano_athena}`
    declare + : [Nat Nat] -> Nat `\label{ln:declare_peano_athena}`
    define [N M] := [?N:Nat ?M:Nat]
    assert* (N + zero = N)
    assert* (N + s M = s (N + M))
}
\end{lstlisting}
\end{minipage}
\end{figure}

\paragraph{Sorts.}
Through $\mathrm{tr}_S$, Maude sorts are translated into Athena domains or datatypes. Whenever a sort $s \in S$ satisfies
that it does not appear in the subsort relation $\leq$ and that there are constructors
with target sort~$s$, $\tr$ translates it to a datatype. In any other case, sorts are translated 
to domains. For instance, \maudecmd{zero} and \maudecmd{s_} in
\fref{lst:peano_maude_simple} are function symbols (with target sort
\maudecmd{Nat}) defined with the \maudecmd{ctor} attribute. Therefore,
$\mathrm{tr}_S$ uses their sort signatures to define the corresponding Athena
datatype, as shown in \fref{ln:datatype_peano_athena} of
\fref{lst:peano_athena_simple}.
The remaining defined function symbols with target sort $s$ are
translated with the \athenacmd{declare} keyword.  This is the case for
the \maudecmd{_+_} function symbol, which is translated into
\fref{ln:declare_peano_athena} of \fref{lst:peano_athena_simple}.

The treatment is different when the sort $s$ is part of the subsort
relation, as is the case in \fref{lst:Peano_maude}. This
module introduces \maudecmd{NzNat} and \maudecmd{Even} as subsorts of
\maudecmd{Nat}. In this case, the previous translation case no longer applies.
For sorts related through the subsort relation (\maudecmd{NzNat},
\maudecmd{Even}, and \maudecmd{Nat})
the translation function $\mathrm{tr}_S$ instantiates each sort as a plain
Athena domain (\fref{ln:doamins_peano} in
\fref{lst:peano_athena}). The process then proceeds through the four remaining
translation components ($\mathrm{tr}_F$, $\mathrm{tr}_T$, $\mathrm{tr}_E$, $\mathrm{tr}_M$):
\begin{enumerate*}[label=(\arabic*)]
  \item function symbols and introducing explicit cast function symbols,
  \item terms using casting function symbols,
  \item equations into Athena assertions using the transformed terms, and
  \item conditional membership axioms using term translation and predicate introduction.
\end{enumerate*}

\paragraph{Function Symbols.}
In the simple module (\fref{lst:peano_maude_simple}), no subsort
relation exists, so $\mathrm{tr}_F$ declares each function symbol
directly: \maudecmd{_+_} becomes the Athena declaration in
\fref{ln:declare_peano_athena} of \fref{lst:peano_athena_simple}.
In the richer module (\fref{lst:Peano_maude}), the subsorts introduce
overloading. The three operators \maudecmd{zero}, \maudecmd{s_}, and
\maudecmd{_+_} are grouped by $\mathrm{tr}_F$ into argument
compatibility classes; for each class the unique maximal
representative is selected and declared in $\gamma$
(\lnrange{ln:maximal_op_peano_start}{ln:maximal_op_peano_end} of
\fref{lst:peano_athena}). Additionally, the two subsort pairs
$\maudecmd{Even} \leq \maudecmd{Nat}$ and
$\maudecmd{NzNat} \leq \maudecmd{Nat}$ produce two cast symbols
(\lnrange{ln:casting_peano_start}{ln:casting_peano_end} of
\fref{lst:peano_athena}).

\paragraph{Terms and Equations.}
In the simple module, no subsort relation exists, so $\mathrm{tr}_T$
leaves terms unchanged and $\mathrm{tr}_E$ translates each equation directly
into an Athena assertion.  In the richer module, implicit coercions must
be made explicit.  Consider the first equation in
\fref{lst:Peano_maude} (\fref{ln:peano-eq-star}): because the
argument has sort \maudecmd{Even} while the operator expects
\maudecmd{Nat}, $\mathrm{tr}_T$ inserts a cast
(\fref{ln:term_eq_peano_strat} of \fref{lst:peano_athena}).
The second equation (\fref{ln:peano-eq-end} of
\fref{lst:Peano_maude}) requires casts on both the argument and the
result (\fref{ln:term_eq_peano_end} of \fref{lst:peano_athena}).
Because the subsort poset of this module has no composable triple
$s \leq s' \leq s''$, no core equalities are generated; in
specifications with longer subsort chains, $\mathrm{tr}_E$ would add
the corresponding transitivity assertions to $\beta$.

\paragraph{Memberships.}
Membership judgements have no native counterpart in Athena's many-sorted logic.
The component $\mathrm{tr}_M$ bridges this gap by introducing a unary predicate
for each sort that appears in a membership axiom; the predicate's domain is
the kind of that sort, so it can be applied to any term of the right kind.
In the richer Peano module, the conditional membership in
\fref{ln:peano-cmb} of \fref{lst:Peano_maude} is translated into
an implication guarded by a membership predicate
(\lnrange{ln:cmb_peano_start}{ln:cmb_peano_end} of
\fref{lst:peano_athena}).
Here, $\mathrm{tr}_M$ chooses the kind of the target sort as the
predicate's domain (\ie the top-most supersort in the connected
component).
The simple module contains no membership axioms, so $\mathrm{tr}_M$
produces no output for it.
The complete translation of \fref{lst:Peano_maude} is shown in
\fref{lst:peano_athena}.

\begin{lstlisting}[language=Athena, style=athena, caption={Peano Athena Module}, 
                 label={lst:peano_athena}]
module PEANO{ 
	domains NzNat Even Nat `\label{ln:doamins_peano}`
	declare s : [Nat] -> NzNat `\label{ln:maximal_op_peano_start}`
	declare zero : [] -> Even
	declare + : [Nat Nat] -> Nat `\label{ln:maximal_op_peano_end}`
	declare Cast_Even_to_Nat : [Even] -> Nat `\label{ln:casting_peano_start}`
	declare Cast_NzNat_to_Nat : [NzNat] -> Nat `\label{ln:casting_peano_end}`
	define [N M] := [?N:Nat ?M:Nat]
	assert* ((+ N (Cast_Even_to_Nat zero)) = N) `\label{ln:term_eq_peano_strat}`
	assert* ((+ N (Cast_NzNat_to_Nat (s M))) 
				= (Cast_NzNat_to_Nat (s (+ N M)))) `\label{ln:term_eq_peano_end}`
	declare is_even : [Nat] -> Boolean `\label{ln:cmb_peano_start}`
	assert* ((is_even N) ==> (is_even s s N)) `\label{ln:cmb_peano_end}` }

\end{lstlisting}

Finally, the correctness of the translation is stated with respect to
equational provability.

\begin{theorem}[Equational Provability]
\label{theo:correquality}
  Let $\ecal = (\Sigma, E\cup A)$ be a membership equational
  theory. If $\Sigma$ is strictly sensible and $\ecal$ is sufficiently
  complete w.r.t.\ $\Omega \subseteq \Sigma$, then for any
  $t,u \in T_{\Sigma,k}(X)$ of the same kind~$k$, the following two
  sentences are equivalent:
  \begin{enumerate}
  \item $t =_\ecal u$
  \item $\mathrm{tr}_T(t) =_{\beta} \mathrm{tr}_T(u)$, \ie
    $\mathrm{tr}_T(t)$ and $\mathrm{tr}_T(u)$ are provably equal
    in Athena under the theory $\mathrm{tr}(\ecal) = (\beta,\gamma)$
    with the inference rules of many-sorted equational logic.
  \end{enumerate}
\end{theorem}

\begin{proof}
  The equational fragment of the translation (components
  $\mathrm{tr}_S$, $\mathrm{tr}_F$, $\mathrm{tr}_T$, and $\mathrm{tr}_E$)
  instantiates the strictly sensible order-sorted--to--many-sorted
  mapping~\cite{li2018method}. The membership component
  $\mathrm{tr}_M$ introduces predicate symbols and assertions of the form
  $\mathrm{is\_s}(\mathrm{tr}_T(t))$ and
  $\mathrm{tr}_T(P) \Rightarrow \mathrm{is\_s}(\mathrm{tr}_T(t))$;
  because $\mathrm{is\_s}$ is a predicate symbol, no
  assertion produced by $\mathrm{tr}_M$ can have the form
  $t' = u'$ for terms $t',u'$ of a non-Boolean sort, so $\mathrm{tr}_M$
  does not introduce new equalities between translated terms. Hence, the
  equational content of $\beta$ is precisely the image of $E$ under
  $\mathrm{tr}_E$, the structural axioms of $A$ under $\mathrm{tr}_F$,
  together with the core equalities.

  \begin{itemize}[label=$\diamond$, leftmargin=10pt]
    \item $t =_\ecal u \Rightarrow \mathrm{tr}_T(t) =_{\beta} \mathrm{tr}_T(u)$

    The proof proceeds by structural induction on the derivation of
    $t =_\ecal u$, case-splitting on the last inference rule applied.

    \begin{itemize}[leftmargin=10pt]
      \item \textbf{Axiom instantiation.}
            Suppose $t =_\ecal u$ is obtained by applying a substitution~$\theta$
            to an equation $l = r$ in $E \cup A$, so that $t = \theta(l)$ and
            $u = \theta(r)$. By $\mathrm{tr}_E$, the translated equation
            $\mathrm{tr}_T(l) = \mathrm{tr}_T(r)$ is in $\beta$, universally
            quantified over all free variables. Because $\mathrm{tr}_T$ is defined
            compositionally and the cast
            insertion at each argument position depends only on the sort of the subterm
            and the expected sort of the position, $\mathrm{tr}_T$ commutes with
            substitution: for each variable $x$, define
            $\hat\theta(x) = \mathrm{tr}_T(\theta(x))$; then
            $\mathrm{tr}_T(\theta(t')) = \hat\theta(\mathrm{tr}_T(t'))$ for any
            term~$t'$, possibly up to core equalities in $\beta$ that identify
            different casting paths to the same target sort. Hence,
            $\hat\theta(\mathrm{tr}_T(l)) =_\beta \hat\theta(\mathrm{tr}_T(r))$
            follows by instantiation of the universally quantified axiom in $\beta$.

      \item \textbf{Reflexivity, symmetry, transitivity.}
            These are inference rules of many-sorted equational logic and
            thus hold natively in Athena's deductive system.

      \item \textbf{Replacement.}
            Suppose $t = C[t']$ and $u = C[u']$ for a context~$C$ and
            $t' =_\ecal u'$, and assume by the induction hypothesis that
            $\mathrm{tr}_T(t') =_\beta \mathrm{tr}_T(u')$.
            The translation $\mathrm{tr}_T$ maps $C$ to a context $C'$ in
            the many-sorted signature~$\gamma$: each function symbol in $C$ is
            replaced by its maximal representative $f^*$, and explicit casts are
            inserted at every argument position where a subsort coercion is needed.
            Because each cast symbol $\mathrm{Cast}_{s \to s'}$ is a
            function symbol in $\gamma$, the standard congruence rule of
            many-sorted equational logic guarantees that if
            $a =_\beta b$ then $\mathrm{Cast}_{s \to s'}(a) =_\beta
            \mathrm{Cast}_{s \to s'}(b)$. Moreover, if $t'$ and $u'$ require different
            casting chains to reach the expected sort of the hole in $C'$, the core
            equalities (added by $\mathrm{tr}_E$ for every composable
            pair $s \leq s' \leq s''$) ensure that both chains yield
            equal results in $\beta$. Hence,
            $\mathrm{tr}_T(C[t']) = C'[\mathrm{tr}_T(t')]
              =_\beta C'[\mathrm{tr}_T(u')] = \mathrm{tr}_T(C[u'])$
            by congruence in many-sorted equational logic.
    \end{itemize}

    Since every inference rule used in the derivation of $t =_\ecal u$
    is preserved, $\mathrm{tr}_T(t) =_\beta \mathrm{tr}_T(u)$.

    \item $\mathrm{tr}_T(t) =_{\beta} \mathrm{tr}_T(u) \Rightarrow t =_\ecal u$

    It suffices to show that $\beta$ is a conservative extension of the
    image of~$E \cup A$ with respect to equalities between translated terms.
    The axioms of~$\beta$ fall into three classes:
    \begin{enumerate}[label=(\alph*)]
      \item translated equations $\mathrm{tr}_T(l) = \mathrm{tr}_T(r)$ for
            each $l = r$ in $E \cup A$;
      \item core equalities
            $\mathrm{Cast}_{s' \to s''}(\mathrm{Cast}_{s \to s'}(x))
            = \mathrm{Cast}_{s \to s''}(x)$ for every composable pair
            $s \leq s' \leq s''$;
      \item membership assertions produced by $\mathrm{tr}_M$.
    \end{enumerate}

    As argued above, class (c) consists entirely of Boolean-sorted
    sentences and cannot produce new equalities between non-Boolean
    terms.
     
    For classes (a) and (b), the argument proceeds as follows:
    \begin{itemize}[leftmargin=10pt]
      \item \textbf{Unambiguous operator mapping.}
            Strict sensibility guarantees that each argument compatibility
            class $[f]_{\mathrm{ac}}$ has a unique maximal representative $f^*$.
            The translation maps every $f$ in the same class to the
            \emph{same} $f^*$ and \emph{omits} all others. Because the
            selection is unique, no two originally distinct function symbols
            are merged by the translation, therefore no invalid equalities between
            terms arise from the operator mapping.

      \item \textbf{Conservative cast equalities.}
            Each cast symbol $\mathrm{Cast}_{s \to s'}$ is a fresh
            uninterpreted function symbol in $\gamma$, distinct for each pair
            $s < s'$.  The only axioms governing cast symbols are the core
            equalities in class~(b), which encode the transitivity and path
            independence of the subsort relation.
            Because the subsort relation $\leq$ is already part of $\ecal$, the core
            equalities do not identify terms that are not already related by
            subsorting.

      \item \textbf{Isomorphism of initial algebras.}
            According to~\citet{li2018method}, the strictly sensible mapping
            combined with the core casting equalities induces a bijection
            between the equivalence classes of the initial order-sorted
            algebra $T_{\Sigma/(E \cup A)}$ and those of the initial
            many-sorted algebra $T_{\gamma/\beta_{\mathrm{eq}}}$, where
            $\beta_{\mathrm{eq}}$ denotes the equational part of $\beta$
            (classes~(a) and~(b)). The sufficient completeness assumption
            ensures that every ground term reduces to a constructor term,
            which is needed for the bijection to be an isomorphism of
            algebras.
    \end{itemize}
     
    Therefore, if $\mathrm{tr}_T(t) =_\beta \mathrm{tr}_T(u)$, then
    $\mathrm{tr}_T(t)$ and $\mathrm{tr}_T(u)$ lie in the same
    equivalence class of $T_{\gamma/\beta_{\mathrm{eq}}}$. By the
    isomorphism, $t$ and $u$ belong to the same class in
    $T_{\Sigma/(E \cup A)}$, \ie $t =_\ecal u$.
  \end{itemize}
\end{proof}



\section{Inductive Reasoning}
\label{sec:induction}

The translation function $\tr$ from $(\Sigma, E \cup A)$ to $(\beta, \gamma)$
presented in \fref{sec:translation} enables equational reasoning over
Maude specifications within Athena by preserving the underlying
equational theory in a many-sorted first-order logic setting.
However, equational reasoning alone is often insufficient to establish
properties of interest, which typically require some form of inductive
reasoning.
This section explains how equational reasoning with $\tr(\ecal)$ in
Athena can be extended with inductive reasoning based on parametric
induction principles. The key idea is that for the induction principle to
be applied, it must be instantiated using the structure of sorts in
$\Sigma$ associated with each connected component in the poset
$(S, \leq)$.
To illustrate this approach, this section presents a concrete
instantiation of such a parametric induction principle, showing how
structural induction can be recovered and applied in practice.

In Athena, there is native support for structural, constructor-based
induction over datatypes, which allows properties to be proved by
reasoning directly on the inductive structure generated by datatype
constructors. However, the translation function $\tr$ systematically
represents many Maude sorts as domains, rather than datatypes, in
order to preserve the structure induced by subsorting in the
many-sorted setting of Athena. As a consequence, Athena's built-in
induction mechanisms are no longer directly applicable, and native
datatype induction is insufficient for reasoning about the translated
specifications.
To address this limitation, the approach adopted here is to equip
Athena with induction principles defined over domains, rather than
relying solely on datatype induction. These principles are formulated
in a parametric way, allowing different induction schemes to be
instantiated depending on the structure of the sorts under
consideration. In particular, this includes natural extensions of
Athena's native structural induction to domains that arise from
translated equational specifications.

More generally, an induction principle in this setting is realized by
declaring a primitive method in Athena that explicitly encodes both
the basis cases and the inductive cases associated with a given sort
structure $(S, \leq)$.
For a selected induction principle $\eta$, the primitive method is
defined separately for each connected component of the sort hierarchy,
reflecting the constructors and subsort relations that characterize
that component. The basis cases correspond to the minimal elements or
constructors of the component, while the inductive cases capture how
the property is preserved by the relevant constructor
applications. Once declared, such a primitive method can be applied
uniformly to predicates over the corresponding domain, thereby
extending equational reasoning with inductive reasoning in a
principled and modular way.

The following definition extends the translation function (\fref{def:translation}) with a
sixth component, $\mathrm{tr}_\eta$, that formalizes the generation of structural
induction principles. It operates on the execution environment~$\rho$---previously noted
as orthogonal to the generated representation $(\beta,\gamma)$, but now required to host
the primitive methods that realize induction.

\begin{definition}[Induction Method Translation $\mathrm{tr}_\eta$] \label{def:induction_translation}
  Let $\ecal$, $\Sigma$, $\Omega$, and $(S,\leq)$ be as in \fref{def:translation}. The
  translation function $\tr$ is extended with:

  \begin{enumerate}[label=\textbf{(\roman*)}] \setcounter{enumi}{5}
    \item \textbf{Induction method translation
      ($\mathrm{tr}_\eta$).}
      For each kind $k$ in $(S,\leq)$, let
      $C = \{\, s \in S \mid
      \exists\, c \in \Omega.\;
      \mathrm{target}(c) = s \,\}
      \cap [k]_{\leq}$
      be the set of sorts in the connected component of
      $k$ that are directly generated by
      $\Omega$-constructors.
      Define the \emph{effective constructor set}
      \[
        \Omega_C^{+} \;=\; \Omega_C \;\cup\;
        \{\, \mathrm{Cast}_{s \to s'} \mid
          s \notin C,\; s' \in C,\;
          s \leq s' \,\},
      \]
      where
      $\Omega_C =
      \{ c \in \Omega \mid \mathrm{target}(c) \in C \}$
      and the second component adds subsort injection
      casts from sorts outside $C$ into $C$.
      By construction, $\Omega_C^{+}$ is jointly
      exhaustive: every constructor normal form of kind
      $k$ is headed by some element of
      $\Omega_C^{+}$.
      For each
      $c : s_1 \times \cdots \times s_n \to s$ in
      $\Omega_C^{+}$, let $I_c =
      \{\, i \mid 1 \leq i \leq n,\;
      s_i \in C \,\}$ be the set of
      \emph{recursive argument positions}.
      The component $\mathrm{tr}_\eta$ generates a
      primitive method that takes a predicate
      $P : \mathrm{tr}_S(k) \to \mathtt{Boolean}$
      and encodes the obligations of
      $\eta_C(P)$ as follows.

    \begin{enumerate}[label=(\alph*)]
      \item \emph{Base-case sentences.}
        \begin{enumerate}[label=(\arabic*)]
          \item For each constant $c : \to s$ in $\Omega_C^{+}$ (0-arity, so $I_c = \varnothing$),
            define
            \[
              b_c \;:=\;
              P\!\bigl(
                \mathrm{Cast}_{s \to k}(c)
              \bigr).
            \]

          \item For each constructor
            $c : s_1 \times \cdots \times s_n \to s$
            in $\Omega_C^{+}$ with $n \geq 1$ and
            $I_c = \varnothing$, define
            \[
              b_c \;:=\;
              \forall\, x_1 \!\cdots\, x_n
              \;.\;
              P\!\bigl(
                \mathrm{Cast}_{s \to k}
                (c(x_1,\ldots,x_n))
              \bigr),
            \]
            where each
            $x_i : \mathrm{tr}_S(s_i)$.
        \end{enumerate}

      \item \emph{Inductive-case sentences.}
        For each constructor
        $c : s_1 \times \cdots \times s_n \to s$ in $\Omega_C^{+}$ with $I_c \neq \varnothing$,
        define the sentence
        \[
          h_c \;:=\;
          \forall\, x_1 \!\cdots\, x_n \;.\;
          \Bigl(\!\bigwedge_{i \in I_c}
          P\!\bigl(\mathrm{Cast}_{s_i \to k}(x_i)\bigr)
          \Bigr)
          \;\Rightarrow\;
          P\!\bigl(\mathrm{Cast}_{s \to k}
            (c(x_1,\ldots,x_n))\bigr),
        \]
        where each $x_i : \mathrm{tr}_S(s_i)$. When $|I_c| = 1$, the conjunction reduces to a
        single hypothesis.

      \item \emph{Primitive method assembly.}
        Let $B = \{b_c\}$ and $H = \{h_c\}$ be the collected base-case and inductive-case
        sentences, and let $(o_1, \ldots, o_m)$ be an enumeration of all sentences in
        $B \cup H$. The primitive method for component~$C$ is added to the execution
        environment $\rho$ as
        \[
          \rho \;\cup\;
          \bigl\{\,
            \mathtt{prim\text{-}method}_{C}(P)
            \;:\;
            \mathtt{check}_m
            \;\Rightarrow\;
            \forall\, x : \mathrm{tr}_S(k).\; P(x)
          \,\bigr\}
          \;\mapsto\; \rho,
        \]
        where each obligation is verified sequentially through nested $\mathtt{check}$
        expressions: $\mathtt{check}_1$ tests $\mathtt{holds?}(o_1)$ and, upon success,
        enters $\mathtt{check}_2$, which tests $\mathtt{holds?}(o_2)$, and so on. In general,
        for $i = 1, \ldots, m$, $\mathtt{check}_i$ tests $\mathtt{holds?}(o_i)$ and proceeds
        to $\mathtt{check}_{i+1}$; when $i = m$, the innermost check concludes
        $\forall\, x : \mathrm{tr}_S(k).\; P(x)$. If any $\mathtt{holds?}(o_i)$ fails, the
        method signals an error.
    \end{enumerate}

    The full translation thus maps $\ecal \mapsto (\beta, \gamma, \rho)$, extending the
    original pair with the execution environment populated by the generated primitive
    methods.
  \end{enumerate}
\end{definition}

As an example, consider the primitive method
in \fref{lst:athena_primitive_induction}. It
defines \athenacode{nat-induction}, a structural induction principle
for the \athenacmd{Nat} domain in \fref{lst:peano_athena}.
It explicitly captures the two components of a standard induction
proof. The first definition, labeled
\athenacode{basis}, corresponds to the sentence of the basis case over the constant
constructor \athenacode{Cast_Even_to_Nat zero}. The second definition, labeled~
\athenacode{ic}, encodes the
inductive case: it states that if the property holds for an
arbitrary \athenacmd{Nat} element \athenacode{x} (inductive
hypothesis), then it needs to hold for its successor,
\athenacode{Cast_NzNat_to_Nat (s x)}.
The \athenacode{check} expression enforces both conditions. It first
verifies that the basis case holds, and if so, proceeds to check the
inductive case. This mechanism effectively restores structural
induction for domains translated from Maude datatypes, allowing the
proof of properties such as associativity or commutativity
complementing the built-in \athenacode{by-induction} method.

\begin{athena}[caption={Primitive induction method for Nat.}, label={lst:athena_primitive_induction}]
primitive-method (nat-induction property) :=
  let {
    basis := (property (Cast_Even_to_Nat zero));
    ic := (forall x (if (property x) (property (Cast_NzNat_to_Nat (s x)))))
    }
    check { (holds? basis) =>
              check {(holds? ic) => (forall x (property x))
                    | else => (error "Inductive step does not hold.")}
              | else => (error "Basis step does not hold.")}
\end{athena}

The following theorem establishes the soundness of inductive proofs carried out using the
primitive methods generated by $\mathrm{tr}_\eta$.

\begin{theorem}[Soundness of Inductive Proofs]
  \label{theo:soundness_induction}
  Let $\ecal = (\Sigma, E \cup A)$ be a membership equational theory that is sufficiently
  complete a connected component of $(S,\leq)$ with kind $k$, and $P : \mathrm{tr}_S(k) \to \mathtt{Boolean}$ a
  predicate. If the primitive method $\mathtt{prim\text{-}method}_C(P)$ generated by
  $\mathrm{tr}_\eta$ (\fref{def:induction_translation}) checks all base cases and
  inductive steps and concludes $\forall\, x : \mathrm{tr}_S(k).\; P(x)$, then for every 
  ground term $t$ of sort $s \in C$ in $T_{\Sigma,k}$, the property $P(\mathrm{tr}_T(t))$
  holds in Athena under the theory $\mathrm{tr}(\ecal) = (\beta,\gamma,\rho)$.
\end{theorem}

\begin{proof}
  Assume, for contradiction, that there exists a ground
  term~$t$ of sort $s \in C$ in $T_{\Sigma,k}$ such that
  $P(\mathrm{tr}_T(t))$ does not hold, even though all
  obligations of
  $\mathtt{prim\text{-}method}_C(P)$ have been checked.
  Because $\ecal$ is sufficiently complete, $t$ reduces to a
  constructor normal form headed by some $c \in \Omega$ with
  $\mathrm{target}(c)$ in the connected component of~$k$.  If
  $\mathrm{target}(c) \in C$, then $c \in \Omega_C \subseteq
  \Omega_C^{+}$; otherwise $c$ is an injection from an external sort
  $s \notin C$ with $s \leq s'$ for some $s' \in C$, so $c =
  \mathrm{Cast}_{s \to s'}$ belongs to the second component of
  $\Omega_C^{+}$.  In either case $c \in \Omega_C^{+}$.  By
  \fref{theo:correquality}, $P(\mathrm{tr}_T(t))$ holds if and only if
  $P$ holds on the translated normal form, so it may be assumed
  without loss of generality that $t = c(t_1,\ldots,t_n)$ is already
  in constructor normal form.
  It proceeds by induction on the number of constructor
  applications in~$t$.

  \begin{description}[leftmargin=0pt]
    \item[Base case:] ($I_c = \varnothing$).
      The constructor $c$ has no argument positions of sorts in $C$, so $t$ contains exactly
      one constructor application. The primitive method verified the base-case sentence
      $b_c$ (\fref{def:induction_translation}). Whether $c : \to s$ is a constant (in which
      case $b_c$ directly asserts $P(\mathrm{Cast}_{s \to k}(c))$) or
      $c : s_1 \times \cdots \times s_n \to s$ with $n \geq 1$ and $I_c = \varnothing$ (in
      which case $b_c$ is universally quantified and instantiates at $\mathrm{tr}_T(t_1),\ldots,
      \mathrm{tr}_T(t_n)$) the verified sentence $b_c$ yields $P(\mathrm{tr}_T(t))$.

    \item[Inductive step:]
      ($I_c \neq \varnothing$). For each $i \in I_c$, the subterm~$t_i$ is a ground term
      of sort $s_i \in C$ with strictly fewer constructor applications than $t$. By the
      induction hypothesis, $P(\mathrm{tr}_T(t_i))$ holds for every $i \in I_c$. The
      primitive method verified the inductive-case sentence $h_c$
      (\fref{def:induction_translation}), which asserts that the conjunction of
      $P(\mathrm{Cast}_{s_i \to k}(x_i))$ for $i \in I_c$ implies $P(\mathrm{Cast}_{s \to k}(
      c(x_1,\ldots,x_n)))$.  Instantiating at $\mathrm{tr}_T(t_1),\ldots,
      \mathrm{tr}_T(t_n)$ and applying modus ponens with the induction hypotheses yields
      $P(\mathrm{tr}_T(t))$.
  \end{description}

  \noindent
  In both cases $P(\mathrm{tr}_T(t))$ holds, contradicting the assumption. Since $t$ was
  arbitrary, $P(\mathrm{tr}_T(t))$ holds for every ground term of kind $k$.
\end{proof}



\section{Case Study}
\label{sec:example}

The case study verifies inductive properties of a compiler for numeric
expressions on a Stack Machine~\cite{arkoudas17}.  The case study
relies heavily on order-sorted features of data types and includes
structural axioms.\footnote{The complete Maude specification of the
compiler with the translated equationally equivalent Athena program
(with proofs) is available at:
\url{https://github.com/FLAGlab/Maude2Athena}} \fref{lst:toy_compiler}
shows the Maude specification, defining three sorts: a source language
of arithmetic expressions (\maudecmd{Exp}), a target instruction set
(\maudecmd{Instr}/\maudecmd{Program}), and a stack machine
(\maudecmd{Stack}).

\begin{maude}[label={lst:toy_compiler}, caption={Toy compiler in Maude (shorten).}]
fmod TOY-COMPILER is
  protecting INT .
  sort Exp .
  subsort Int < Exp .
  op _plus_ : Exp Exp -> Exp [ctor] .
  op _minus_ : Exp Exp -> Exp [ctor] .
  op _mult_ : Exp Exp -> Exp [ctor] .
  sorts Instr Program Stack .
  subsort Instr < Program .
  op push : Int -> Instr [ctor] .
  ops add sub mult : -> Instr [ctor] .
  op nil : -> Program [ctor] .
  op _++_ : Program Program -> Program [ctor assoc id: nil] .
  op empty : -> Stack [ctor] .
  op _::_ : Int Stack -> Stack [ctor] .
  op I : Exp -> Int .         --- interpreter
  op compile : Exp -> Program .
  op exec : Program Stack -> Stack .
  vars N N1 N2 : Int .  vars E1 E2 : Exp .
  var P : Program .     var S : Stack .
  eq I(N) = N .
  eq I(E1 plus E2) = I(E1) + I(E2) .
  --- I equations for minus, mult analogous
  eq compile(N) = push(N) . `\label{ln:compile_push}`
  eq compile(E1 plus E2) =
     compile(E2) ++ compile(E1) ++ add .
  --- compile equations for minus, mult analogous
  eq exec(nil, S) = S .
  eq exec(push(N) ++ P, S) = exec(P, N :: S) . `\label{ln:exec_push}`
  eq exec(add ++ P, N1 :: N2 :: S) =
     exec(P, (N1 + N2) :: S) .
  --- exec equations for sub, mult analogous
  ...
endfm
\end{maude}

The compiler specification relies on two order-sorted features. First, \maudecmd{Int < Exp} lets
integers appear directly as expressions without an explicit injection constructor.
Second, \maudecmd{Instr < Program} lets single instructions be used where programs are
expected. The operator \maudecmd{_++_} carries the structural axioms \maudecmd{assoc}
and \maudecmd{id: nil}.

To preserve the semantics of \maudecmd{Int < Exp}, the translation introduces explicit
cast functions for every subsort relation in the sort poset, as explained in
\fref{sec:translation}. \mtoa automatically inserts these casts.  For example, the
equations in lines \fref{ln:compile_push} and \fref{ln:exec_push}
are translated to:

\begin{athenainline}
    declare push : [Int] -> Instr
    declare compile : [Exp] -> Program
    declare Cast_Int_to_Exp : [Int] -> Exp
    declare Cast_Instr_to_Program : [Instr] -> Program
    assert* eq_4 := ((compile (Cast_Int_to_Exp N)) 
                        = (Cast_Instr_to_Program (push N)))
    assert* eq_9 := ((exec (++ (Cast_Instr_to_Program (push N)) P) S)
                        = (exec P (:: N S)))
\end{athenainline}

Since Athena requires unique function symbols and does not support
Maude's mixfix notations natively, the tool flattens operators and
resolves overloading.  Specifically, the mixfix Maude operator
\maudecmd{_plus_} is translated to a standard prefix function in
Athena. If overloading collisions are detected, unique identifiers are
generated based on sort signatures:

\begin{athenainline}
    declare plus : [Exp Exp] -> Exp
    declare ++ : [Program Program] -> Program
    assert* ((compile (plus E1 E2)) =
             (++ (compile E2)
                 (++ (compile E1) (Cast_Instr_to_Program add))))
    define [_v1 _v2 _v3 _v4] :=
           [?_v1:Program ?_v2:Program ?_v3:Program ?_v4:Program]
    assert* assoc_++ := ((++ (++ _v1 _v2) _v3) = (++ _v1 (++ _v2 _v3)))
    assert* left_id_++ := ((++ nil _v4) = _v4)
    assert* right_id_++ := ((++ _v4 nil) = _v4)
\end{athenainline}
This listing also illustrates how the structural axioms $A$ are handled. The operator
\maudecmd{_++_} is declared with the attributes \maudecmd{assoc, id: nil}, meaning that
rewriting is performed modulo associativity and identity. Because Athena has no built-in
support for such axioms, $\tr_F$ translates each structural axiom into an explicit
equational assertion: \athenacode{assoc_++} encodes associativity, while
\athenacode{left_id_++} and \athenacode{right_id_++} encode the left and right identity
axioms for \athenacode{nil}. In general, any combination of associativity, commutativity,
and identity (ACU) attributes attached to operators in $A$ is translated by $\tr_F$ into
the corresponding set of universally quantified assertions in $\beta$.

The tool analyzes the signature to determine which sorts can be
modeled as \athenacmd{datatypes} (allowing native induction) and which
must be \athenacmd{domains}.  In this example, the Maude \maudecmd{Stack} is
translated as an Athena datatype:
\begin{athenainline}
    datatype Stack := empty | (:: Int Stack)
\end{athenainline}
However, \maudecmd{Exp} and \maudecmd{Program} are translated as
domains to accommodate the subsorting relations that are otherwise
incompatible with standard datatype constructors.

The result of this translation is a fully typed Athena module where
the correctness of the compiler (\athenacmd{forall e . (exec (compile e)) = (I e)}) can be rigorously proven. The core
verification goal is the following theorem, which asserts that
compiling an expression \athenacmd{e} followed by a program
\athenacmd{p} is equivalent to executing \athenacmd{p} on a stack with
the evaluated result of \athenacmd{e}:\footnote{The correctness property follows
trivially as a corollary of this theorem by instantiating it to the empty program
on an empty stack.}

\begin{athenainline}
    (forall ?p ?s .   (exec (++ (compile ?e) ?p) ?s) = 
                      (exec ?p (:: (I ?e) ?s)))
\end{athenainline}

Note that this sentence cannot be asserted as a logical sentence
because \athenacmd{Exp} is a domain, not a datatype. Instead, it is
defined as a predicate over expressions.  This enables the theorem to
be passed as an argument to the induction method of choice.  In this
case, the predicate \athenacmd{correctness} is defined as:
\begin{athenainline}
    declare correctness: [Exp] -> Boolean
    assert* correctness_def := (iff (correctness ?e)
        (forall ?p ?s . 
            (exec (++ (compile ?e) ?p) ?s) = (exec ?p (:: (I ?e) ?s))))
\end{athenainline}

Proving that \athenacode{(forall ?e . correctness ?e)} holds, brings
to attention the challenges of reasoning over non-inductive domains
and handling explicit subsorting.

Note that in Maude, \maudecmd{Exp} is defined via subsorting ($Int <
Exp$), which prevents it from being translated into a standard Athena
\athenacmd{datatype}. Consequently, Athena's native
\athenacode{by-induction} method cannot be applied. To overcome this,
\mtoa automatically generates a custom \textbf{primitive method}
that reconstructs the inductive principle of the original Maude sort.
The generated method, \athenacode{exp-induction}, takes the
\maudecmd{correctness} predicate as an argument and strictly enforces
that the user proves the inductive step for every constructor (e.g.,
\athenacode{plus}, \athenacode{mult}) before discharging the goal:
\begin{athenainline}
primitive-method (exp-induction property) :=
    let {
        basis_n  := (forall ?n (property (Cast_Int_to_Exp ?n)))
        ic_plus  := (forall ?e1 (forall ?e2  
                        (if (and (property ?e1) (property ?e2)) 
                            (property (plus ?e1 ?e2)))));
        ic_minus := ...
        ic_mult  := ...
    }
    check { (holds? basis_n) =>   ...  }
\end{athenainline}

The translation replaces implicit subsorting with explicit cast
operators. This requires the proof script to handle these casts in the
basis cases. For instance, the basis case for integers is not merely
\athenacode{forall n}, but rather \athenacode{forall n} applied to the
cast \athenacode{Cast_Int_to_Exp}. The proof proceeds by unwrapping
these casts using the axioms generated by the translation (\eg
\athenacode{eq_4}):

\begin{athenainline}
define basis_step := (forall n . correctness (Cast_Int_to_Exp n))
conclude basis_step
  pick-any n:Int p:Program s:Stack
    (!chain [ (exec (++ (compile (Cast_Int_to_Exp n)) p) s)
            = (exec (++ (Cast_Instr_to_Program (push n)) p) s)   [eq_4]
            = (exec p (:: n s))                                  [eq_9]
            ... ])
\end{athenainline}

This explicit handling ensures that the proof is rigorous and
mathematically consistent with the original order-sorted
semantics. Finally, the proof is concluded by invoking the custom
induction method, which aggregates the base case and inductive steps
to discharge the main theorem: \athenacode{(!exp-induction correctness)}.


\section{Related Work}
\label{sec:related_work}

The formal verification of rewriting logic specifications typically
involves a trade off between the high performance execution provided
by environments like Maude and the interactive deductive capabilities
of theorem provers like Athena~\cite{arkoudas17}, Isabelle~\cite{wenzel2008isabelle}, 
or Lean~\cite{moura21}.
There are three main methodologies designed to address this challenge:
\begin{enumerate*}[label=(\arabic*)]
\item the development of native domain-specific theorem provers,
\item the application of automated symbolic reasoning techniques, and
\item the integration with general purpose proof assistants via
  translation.
\end{enumerate*}

The native approach, pioneered by the original Maude
\ac{ITP}~\cite{clavel2007all} and modernized by
NuITP~\cite{duran2024nuitp}, offers powerful automation using
variant-based reasoning directly at Maude's metalevel. While highly
effective for discharging goals automatically without translation,
tools like NuITP inherently generate proofs as metalevel search
traces. In contrast, \mtoa translates specifications into Athena to
leverage its natural-deduction framework, which closely mirrors
standard mathematical reasoning~\cite{arkoudas2005}. Rather than
competing on pure automation, the proposed approach emphasizes proof
transparency and structure. Athena's built-in methods (equality
chaining, structural induction, case analysis, and many others) allow
proofs to be developed and read as structured mathematical arguments
rather than execution traces or low-level tactic scripts, following
human-like proof
strategies~\cite{BRINGSJORD_GOVINDARAJULU_2021}. Consequently, \mtoa
occupies a complementary niche to native provers like NuITP: it
targets verification scenarios where transparent, independently
checkable, and human-readable proofs are the primary requirement.

A highly effective alternative to interactive deduction relies on
automated symbolic methods. Maude integrates with SMT solvers to
discharge proof obligations generated from
specifications.
Furthermore, techniques such as narrowing and its combination with SMT
solving allow for the analysis of infinite state
systems~\cite{rocha2014rewriting}.  Reachability Logic and its
associated tools also provide mechanisms to deductively reason about
Maude programs~\cite{romero2018reachability,arias2023symbolic}.  These
automated methodologies are essential for verifying safety properties
and analyzing state spaces, though they may not cover properties
requiring complex inductive arguments or higher order reasoning.

The third strategy is to translate specifications to an external
formalism to leverage its proof infrastructure. The Heterogeneous Tool
Set (Hets)~\cite{codescu2010integrating} treats Maude as an
institution, enabling translations to Isabelle/HOL via
CASL~\cite{astesiano2002casl}. While theoretically robust, this path
relies on complex encodings to simulate subsorting. Similarly,
approaches to translate order-sorted algebras into many-sorted algebras have been explored to
leverage standard theorem provers~\cite{li2018method,goguen1992order}.
The Maude2Lean framework~\cite{rubio2022maude2lean} presents a direct
translation from rewriting logic to the calculus of inductive
constructions. Unlike the institutional approach of Hets, Maude2Lean
translates the syntax of Maude terms and the semantics of rewriting
logic into inductive datatypes in Lean. This preserves the inductive
structure of the logic, allowing native structural induction on terms
and relations. However, such translations shift the first-order
theorem proving problem to a different domain (e.g., type theory or
higher order logic), sometimes incurring in a verbosity trade-off, as
proofs become arguments about the translated semantics rather than
direct reasoning about the domain model.


\section{Concluding Remarks}
\label{sec:conclusions}

This work presents \mtoa, a tool for translating order-sorted Maude
specifications into the many-sorted first-order logic of Athena. The
main contribution of this work is the reconciliation of symbolic
reasoning and equational rewriting with the power of inductive
reasoning, enabling specifications that leverage the best of both
worlds.
To the best of the authors' knowledge, this is the first
implementation grounded in the theoretical translation proposed
by~\citet{li2018method}. The proposed approach bridges the semantic
gap between the two formalisms by making subsorting explicit via cast
operators, modeling membership as predicates, and recovering structural induction through custom primitive
methods.
The $\mtoa$ framework was validated using several examples. It
includes the case study of a compiler specification, demonstrating the
approach's ability to handle deep inductive structures and operator
overloading, while preserving the equational consistency of the
original system.

The future work focuses on two key areas. First, to enhance modularity
by mapping Maude's built-in modules directly to Athena's native
domains. Such integration would allow developers to leverage Athena's
existing decision procedures and automated reasoning capabilities for
standard data types.
Second, extending \mtoa to handle rewrite rules to enable the
verification of concurrent and distributed systems. The current
version is restricted to functional modules, while rewrite rules are
the primary mechanism for defining concurrent state
transitions. Athena, being grounded in standard first-order logic,
lacks native primitives for modeling the non-determinism and
asynchrony inherent in such rules.  Rather than constructing a new
concurrency framework within Athena from scratch, the idea is to
directly work on top of Talcott's formalization of Actor theories in
rewriting logic~\cite{talcott2002actor}. By translating these
pre-validated semantic structures directly into Athena, Maude's system
modules can be mapped to logical transition relations. On the one
hand, this would significantly reduce the formalization burden.  On
the other hand, it will effectively equip Athena with the capability
to reason about temporal properties and distributed protocols without
needing to manually reconstruct the underlying theory of concurrency.


\paragraph{Acknowledgements.}
The authors would like to thank the anonymous reviewers for their very
insightful comments on an earlier version of this paper.
They would like to also thank Konstantine Arkoudas for his insight
on encoding structural induction over plain domains as a primitive
method in Athena to preserve the inductive reasoning capabilities.
Rocha acknowledges support from the SGR project PROMUEVA (BPIN
2021000100160) under the supervision of Minciencias (Ministerio de
Ciencia Tecnología e Innovación, Colombia).

\printbibliography

\end{document}